\documentclass[11pt,reqno]{amsart}
\usepackage{amsmath,amssymb,amsthm,mathtools,wasysym,calc,verbatim,enumitem,tikz,pgfplots,url,hyperref,mathrsfs,cite,fullpage,bbm}

\usepackage{comment}

\newtheorem{theorem}{Theorem}[section]
\newtheorem{lemma}[theorem]{Lemma}
\newtheorem{corollary}[theorem]{Corollary}
\newtheorem{prop}[theorem]{Proposition}

\newtheorem{claim}[theorem]{Claim}
\newtheorem{fact}[theorem]{Fact}

\theoremstyle{definition}
\newtheorem{defn}[theorem]{Definition}
\newtheorem{assumption}[theorem]{Assumption}
\theoremstyle{remark}

\newcommand\N{\mathbb{N}}
\newcommand\R{\mathbb{R}}
\newcommand\Z{\mathbb{Z}}

\newcommand\cN{\mathcal{N}}

\newcommand{\bv}{\mathbf{v}}

\newcommand\cT{\mathcal{T}}

\newcommand\eps{\varepsilon}
\newcommand{\cS}{\mathcal{S}}

\renewcommand{\leq}{\leqslant}
\renewcommand{\geq}{\geqslant}

\renewcommand{\to}{\rightarrow}

\def\eps{\varepsilon}

\def\C{\mathbb{C}}
\def\R{\mathbb{R}}
	
\def\Z{\mathbb{Z}}
\def\N{\mathbb{N}}

\def\1{\mathbbm{1}}

\newcommand{\one}{\mathbf{1}}
\newcommand{\bx}{\mathbf{x}}
\newcommand{\by}{\mathbf{y}}
\newcommand{\bz}{\mathbf{z}}
\newcommand{\supp}{\mathrm{supp}}

\begin{document}
	
	\title{Quasipolynomial-time algorithms for  Gibbs point processes}
	\begin{abstract}
		We demonstrate a quasipolynomial-time deterministic approximation algorithm for the partition function of a Gibbs point process interacting via a finite-range stable potential.  This result holds for all activities $\lambda$ for which the partition function satisfies a zero-free assumption in a neighborhood of the interval $[0,\lambda]$.  As a corollary, for all finite-range stable potentials we obtain a quasipolynomial-time determinsitic algorithm for all $\lambda < /(e^{B + 1} \hat C_\phi)$ where $\hat C_\phi$ is a temperedness parameter and $B$ is the stability constant of $\phi$.  In the special case of a repulsive potential such as the hard-sphere gas we improve the range of activity by a factor of at least $e^2$ and obtain a quasipolynomial-time deterministic approximation algorithm for all $\lambda < e/\Delta_\phi$, where $\Delta_\phi$ is the potential-weighted connective constant of the potential $\phi$.   Our algorithm approximates coefficients of the cluster expansion of the partition function and uses the interpolation method of Barvinok to extend this approximation throughout the zero-free region.
	\end{abstract}
	
	\author{Matthew Jenssen}
	\address{King's College London, Department of Mathematics}
	\email{matthew.jenssen@kcl.ac.uk}
	
	\author{Marcus Michelen}
	\address{Department of Mathematics, Statistics, and Computer Science\\ University of Illinois at Chicago}
	\email{michelen.math@gmail.com}
	
	\author{Mohan Ravichandran}
	\address{Bogazici University, Department of Mathematics}
	\email{mohan.ravichandran@gmail.com}

	\maketitle

	\section{Introduction}
	
	Gibbs point processes are a fundamental model of random spatial phenomena in the continuum. Most classically, such processes are used to model a gas under local interactions (see Ruelle's \cite{ruelle1999statistical} text).  Beyond that, Gibbs point processes are used to model various physical phenomenon that often exhibit local repulsion, such as the locations of galaxies in the universe, the time and place of earthquakes, and the growth of trees in a forest; see \cite{moller2007modern,daley2007introduction} for these applications and more.  A simple and well-studied example of a Gibbs point process is the \emph{hard-sphere model}, where one samples a Poisson point process in a set of finite volume in $\R^d$ and conditions on no two points having distance less than some parameter $r > 0$.  In order to better understand these models one often wants 
	to approximately compute the partition function of the model, which may be understood as a weighted count of allowable configurations of points.  The partition function grows exponentially in the volume, making exact computation intractable even for basic examples.  Additionally, the rate of exponential growth is equal to the infinite-volume pressure, a central quantity in statistical physics. 
	
	Approximating the partition function and sampling---either approximately or exactly---are the two main algorithmic problems associated to Gibbs point processes.  Under very mild assumptions, polynomial-time approximate sampling of the point process can be used to provide a \emph{randomized} approximation to the partition function. Many techniques have been applied to Gibbs point processes for approximate and exact sampling in certain regimes; in fact, the seminal Metropolis-Hastings algorithm was developed to sample from the hard-sphere model in dimension 2~\cite{metropolis1953equation}.  
	
	On the other hand, \emph{deterministic} approximation algorithms for partition functions of Gibbs point processes are less well-understood. 	For Gibbs point processes, to our knowledge the only rigorous result giving a deterministic algorithm is that of Friedrich, G\"obel, Katzmann, Krejca and Pappik \cite{friedrich2022algorithms}, which shows that for the special case of hard-spheres, one may approximate the partition function in quasipolynomial time for a certain range of parameters. The authors show that the hard-sphere model can be well approximated by its discrete analogue (the hard core model) allowing them to apply known algorithmic results from the discrete setting (see Section~\ref{sec:context} for more detail). 
	
		In this paper, we provide a quasipolynomial time deterministic approximation algorithm for the partition function for a general class of stable Gibbs point processes.  Our main result is stated only under the assumption of zero-freeness of the partition function; from there, we deduce two main corollaries using existing zero-freeness results from the literature, one which applies for all stable potentials and a stronger result that applies for the more restricted class of repulsive potentials.  We defer formal statements to Section~\ref{sec:resultsstatement}.
		
		Our approach is via Barvinok's interpolation method~\cite{barvinok2016combinatorics} combined with use of the cluster expansion for Gibbs point processes. This allows us to work with the Gibbs process directly, rather than a discrete approximation of the process. By combining our result with the zero-free region for stable potentials guaranteed by the cluster expansion, we obtain what appears to be the first algorithmic result for stable, non-repulsive potentials. In the special case of a repulsive potential, combining our result with previous work of the second author and Perkins on zero-freeness \cite{mp-CC} yields the first quasipolynomial time deterministic approximation algorithm for a large class of repulsive potentials (which includes the hard-sphere model) and range of parameters. For the hard-sphere potential $\phi$, we note that \cite[Lemma 12]{mp-CC} gives an explicit bound of $\Delta_\phi < C_\phi(1 - 8^{-d-1})$ where $\Delta_\phi, C_\phi$ denote the potential-weighted connective constant and temperedness constant of $\phi$ respectively (defined in the next sections).  This demonstrates that our algorithm works for a wider range of parameters than the previous deterministic algorithms of~\cite{friedrich2022algorithms}.  Additionally, \cite{friedrich2021spectral} argues that the connective constant of the discretization used in \cite{friedrich2021spectral,friedrich2022algorithms} would not provide an improvement to their results; as such, working in the continuum and using the zero-freeness result of \cite{mp-CC} gets around this issue.  
	
	\subsection{Formal definition of the model}
	
	The point processes we consider are defined by three parameters: 
	\begin{itemize}
	\item a measurable set $S \subset \R^d$ of finite volume,
	\item a parameter $\lambda \geq 0$ referred to as the \emph{activity} or \emph{fugacity},
	\item a \emph{pair potential} $\phi:  \R^d \to \R \cup \{+ \infty\}$ satisfying $\phi(x) = \phi(-x)$. 
	\end{itemize}

	The \emph{temperedness constant} of a potential is defined as \begin{equation}\label{eq:Cphidef}
		C_\phi = \int_{\R^d} |1 - e^{-\phi(x)}|\,dx\,.
	\end{equation}
	The temperedness constant may be understood as a measure of the strength of the potential.  We say that a potential $\phi$ is \emph{tempered} if $C_\phi < \infty$, {and always assume that $\phi$ is tempered}.  The \emph{energy} of a configuration of points $\{x_1,\ldots,x_{N}\} \subset \R^d$ is defined by \begin{equation}
		H(x_1,\ldots,x_N) = \sum_{1 \leq i < j \leq N} \phi(x_i - x_j)\,.
	\end{equation}

	{We will always assume that $\phi$ is \emph{stable}, meaning that there is a constant $B \geq 0$ so that for all $N$ and $x_1,\ldots,x_N$ we have \begin{equation}\label{eq:stable-def}
			H(x_1,\ldots,x_N) \geq - BN\,.
		\end{equation}
The infimum over such $B$ is called the	\emph{stability constant} of $\phi$.  The assumption of stability is used to show that the partition function---and thus the Gibbs point process---itself is well-defined.  Under certain conditions on $\phi$, the assumption of stability is a necessary condition for the point process to be well-defined (see \cite[Section 3.2]{ruelle1999statistical} for a detailed discussion and many examples).}
    {A potential $\phi$ is \emph{repulsive}, if $\phi(x) \geq 0$ for all $x$. In particular, repulsive potentials are stable} {with stability constant $B = 0$}. 
	
	The \emph{Gibbs point process} in $S$ with potential $\phi$ at activity $\lambda$ is the point process in $S$ whose density against the Poisson point process of activity $\lambda$ is proportional to $e^{-H(x_1,\ldots,x_N)}$. 
	The \emph{grand canonical partition function} at activity $\lambda$ is defined by \begin{equation} \label{eq:Z-def}
		Z_S(\lambda) = \sum_{ k \geq 0} \frac{\lambda^k}{k!} \int_{S^k} e^{-H(x_1,\ldots,x_k)}\,dx_1\,\ldots\,dx_k\,.
	\end{equation}
	Throughout, we work with the case of $S = \Lambda_n:= [-n,n]^d \subset \R^d$, i.e.\  the axis parallel box of side-length $2n$. One of the most studied examples of a Gibbs point process is the \emph{hard sphere model} which is defined by the potential
\begin{align}\label{eq:hardpotential}
\phi(x)=
\begin{cases*}
                    +\infty & if  $\|x\|_2 < r$  \\
                    0 & otherwise\, ,
                 \end{cases*} 
\end{align}
for fixed $r>0$. The hard sphere model is supported on configurations $\{x_1,\ldots,x_{N}\}$ such that $H(x_1,\ldots,x_N) =0$ i.e. configurations consisting of the centers of a packing of spheres of radius $r/2$.   
	
	{Another similar example is the \emph{Strauss potential}, in which the $+\infty$ in the definition of the hard sphere potential is replaced with a parameter $a > 0$.  }
	
	{Among the most common examples of a stable potential that is not repulsive is a \emph{Lennard-Jones potential} (see \cite[Section 3.2.10]{ruelle1999statistical}).  While there are many examples of potentials that are called Lennard-Jones potentials, they are characterized by being strongly repulsive at short distances and weakly attractive at far distances.  A large family of widely used potentials is of the form $$\phi(x) = A \|x\|_2^{-2\alpha} - B\|x\|_2^{-\alpha} $$ for parameters $A,B > 0$ and $\alpha > d$, where we recall $d$ is the underlying dimension.  One may also truncate this potential to be of finite range by simply setting it equal to $0$ if $\|x\|_2 \geq T$ for some parameter $T$, and this still yields a stable potential. }

	\subsection{Statement of results}\label{sec:resultsstatement}

	Our results will require only two additional assumptions on the potential: first a basic assumption on its form so that we may approximately compute certain volumes; and second a zero-freeness assumption. {The zero-freeness assumption we use for stable (non-repulsive) potentials is a classical result that follows from the cluster expansion \cite{ruelle1999statistical,procacci2017convergence}, while in the repulsive case we will use} the work of the second author and Perkins~\cite{mp-CC}.  We begin with the assumption required for computational purposes.
	\begin{assumption}\label{assumption:potential}
		There are compact {centrally symmetric} convex sets $K_1 \subset K_2 \subset \ldots \subset K_\ell$ so that the function $x\mapsto \exp(-\phi(x))$ is $L$-Lipschitz on each of $\Delta_j:= K_{j} \setminus K_{j-1}$.  Additionally assume that there is an $R>0$ so that $\supp(\phi) \subset K_\ell \subset [-R,R]^d$ and $[-1/R,1/R]^d \subset K_1$.  	We work in a real-valued model of computation, and assume unit cost for elementary operations; we also assume that evaluation of { $e^{-\phi}$} as well as checking if a point lies in a given $K_j$ have unit cost.
	\end{assumption}
	{
	Assumption \ref{assumption:potential} essentially has two elements: it assumes that the potential $\phi$ has a certain amount of (piecewise) regularity, and also assumes that $\phi$ is of finite range.  We are not aware of any natural potential $\phi$ in the literature that fails to satisfy the piecewise regularity portion of the assumption. }{There are however many natural potentials that are not of finite range.}

	We remark that the hard sphere potential~\eqref{eq:hardpotential}, {the Strauss potential, and the truncated Lennard-Jones potential are} easily seen to satisfy Assumption~\ref{assumption:potential}.\\
	
		\noindent\textbf{A note on asymptotic notation.} Throughout the paper we think of the parameters $d, \ell, L$ and $R$ from Assumption~\ref{assumption:potential} as fixed and allow the implicit constants in our asymptotic notation to depend on these parameters. \\
	
	Our main theorem is that we obtain deterministic approximation algorithms for $Z_S(\lambda)$ for potentials satisfying Assumption~\ref{assumption:potential} and the following zero-freeness assumption on $Z_S(\lambda)$.  
	
	\begin{assumption}\label{assumption:zero-free}
		We say a {stable} potential $\phi$ satisfies Assumption \ref{assumption:zero-free} at $\lambda_0 > 0$ if there exist constants $\delta, C > 0$ so that the following holds.  For all bounded, measurable $S \subset \R^d$ we have \begin{equation}
		Z_S(\lambda) \neq 0 \text{\,\, and\,\, }	\frac{1}{|S|} |\log Z_S(\lambda)| \leq C \qquad \text{ for all }\lambda \in \mathcal{N}_\delta([0,\lambda_0]) \,,
		\end{equation}
		where
			\[
			\mathcal{N}_\delta([0,\lambda_0]) = \{z\in \mathbb C :  d(z, [0,\lambda_0]) < \delta \}\, .
			\]
	\end{assumption}

	Assumption \ref{assumption:zero-free} may be understood as saying that the Gibbs point process exhibits no phase transition in the regime $[0,\lambda_0]$ in the Lee-Yang sense (see Section \ref{sec:context}).
	
	Given a number $Z$, an $\eps$-approximation to $Z$ is a value $\hat{Z}$ so that $e^{-\eps}\hat{Z} \leq Z \leq e^{\eps} \hat{Z}$.  Our main theorem asserts a quasipolynomial-time approximation for the partition function under these two assumptions.  Recall that $\Lambda_n:= [-n,n]^d \subset \R^d$.

	\begin{theorem}\label{th:main}
		Let $\phi$ be a {stable} pair potential that satisfies the regularity Assumption \ref{assumption:potential}. Suppose $\phi$ satisfies the zero-freeness Assumption \ref{assumption:zero-free} at $\lambda \geq 0$.  Let $\eps\in (0,1)$, $n\in \N$. Then there is a deterministic algorithm to compute an $\eps$-approximation to $Z_{\Lambda_n}(\lambda)$ with runtime at most $\exp(O(\log^3(|\Lambda_n|/\eps)))$. 	\end{theorem}
	
	{We note that the implicit constant in the exponent depends on the potential $\phi$ as well as the activity $\lambda$.}
	
	{We will deduce two main corollaries from Theorem \ref{th:main}.  The first of which deduces an algorithm for all stable potentials in the regime of cluster expansion convergence.  The fact that for $\lambda < (e^{1+2B} C_\phi)^{-1}$ a stable potential $\phi$ satisfies Assumption \ref{assumption:zero-free} is a direct consequence of the cluster expansion (see, e.g., \eqref{eq:C_k-def} or \cite[Thm 4.2.3]{ruelle1999statistical}).}
	{ This classical result was recently improved by  Procacci and Yuhjtman~\cite{procacci2017convergence} who showed that any stable tempered potential $\phi$ satisfies Assumption \ref{assumption:zero-free} at all $\lambda < (e^{1+B} \hat C_\phi)^{-1}$, where
	\begin{align}\label{eqHatTemp}
	\hat C_\phi = \int_{\R^d} 1 - e^{-|\phi(x)|}\,dx\,.
	\end{align}
	We note that $\hat C_\phi\leq C_\phi$ for all $\phi$ with equality if $\phi$ is repulsive.
	 }
	
{	\begin{corollary}
			Let $\phi$ be a stable tempered potential satisfying Assumption \ref{assumption:potential}, {let $B$ denote its stability constant and let $\hat C_\phi$ be as in~\eqref{eqHatTemp}.} Let $\eps\in (0,1)$, $n\in \N$. Then for all {$\lambda < (e^{1+B} \hat C_\phi)^{-1}$}, there is a deterministic algorithm to compute an $\eps$-approximation to $Z_{\Lambda_n}(\lambda)$ with runtime at most $\exp(O(\log^3(|\Lambda_n|/\eps)))$. 	
	\end{corollary}
}
	
	{Our second corollary applies for repulsive potentials---i.e.\ the case of $B = 0$---for a much wider range of parameters.}	To apply Theorem \ref{th:main} {in the repulsive case}, we will use pre-existing work on zero-freeness of Gibbs point process partition functions.  In \cite{mp-CC}, the potential-weighted connective constant $\Delta_\phi$ was introduced, which captures a relationship between the strength of the potential and the geometry of the underlying space.
	
	First, define $V_k$ via \begin{equation}
		V_k = \int_{(\R^d)^k} \prod_{j = 1}^k \left[ \exp\left( - \sum_{i = 0}^{j - 2} \one_{\| v_j - v_i\| < \|v_i - v_{i+1}\| } \phi(v_j - v_i) \right) \cdot (1 - e^{-\phi(v_j - v_{j-1})}) \right]\,d\mathbf{v}
	\end{equation}  
	where we write $d\mathbf{v} = dv_1\,dv_2\,\ldots \,d v_k$ and interpret $v_0 = 0$, and in the case of $j = 1$ interpret the empty sum as equal to $0$.  Since the potential $\phi$ is repulsive, the sequence $\{V_k\}_{k \geq 1}$ is submultiplicative and so we may define the \emph{potential-weighted connective constant} $\Delta_\phi$ via \begin{equation}
		\Delta_\phi = \lim_{k \to \infty} V_k^{1/k} = \inf_{k \geq 1} V_k^{1/k}\,.
	\end{equation}
	
	For any repulsive potential $\phi$ we have $\Delta_\phi \leq C_\phi$ where $C_\phi$ is the temperedness constant defined at~\eqref{eq:Cphidef}. So long as $\phi$ is non-trivial we in fact have $\Delta_\phi < C_\phi$.  In \cite{mp-CC} it was shown that Assumption \ref{assumption:zero-free} is satisfied for each tempered repulsive potential for all $\lambda_0 < e/\Delta_\phi$.  
	We therefore obtain the following immediate corollary.   
	
	\begin{corollary}
		Let $\phi$ be a repulsive tempered potential satisfying Assumption \ref{assumption:potential} and let $\Delta_\phi$ denote its potential-weighted connective constant. Let $\eps\in (0,1)$, $n\in \N$. Then for all $\lambda < e/\Delta_\phi$, there is a deterministic algorithm to compute an $\eps$-approximation to $Z_{\Lambda_n}(\lambda)$ with runtime at most $\exp(O(\log^3(|\Lambda_n|/\eps)))$. 	\end{corollary}

	\subsection{Context and related work}\label{sec:context}
	
	Much of the classical work on Gibbs point processes has consisted of showing the absence of a phase transition when $\lambda$ is small.  There are many  definitions and notions of a phase transition.  One of the most robust definitions is due to Lee and Yang \cite{lee1952statistical,yang1952statistical}, which states that a phase transition is a point at which the pressure fails to be analytic.  Often, various other definitions of phase transitions, e.g. in terms of infinite-volume Gibbs measures, can be shown to coincide with the Lee-Yang definition (see, e.g., \cite{dobrushin1985completely} for some rigorous equivalences in the discrete case).  
	
	It remains a major open problem to demonstrate---or rule out---a phase transition for even a single non-trivial pair potential $\phi$\footnote{The first example of a continuous system for which a phase transition was proven was the Widom-Rowlinson model, a result due to Ruelle\cite{ruelle1971widomrowlinson}. However this does not fit into the framework of indistinguishable particles that we consider in this paper.}.  Groenveld \cite{groeneveld1962two} used the cluster expansion to prove that no phase transition occurs for $\lambda < 1/(eC_\phi)$ for repulsive potentials; this was extended to the broader class of stable potentials by Penrose \cite{penrose1963convergence} and Ruelle \cite{ruelle1963correlation}.  For repulsive potentials, works of Meeron \cite{meeron1962indirect,meeron1970bounds} extends this regime to $1/C_\phi$.  Taking inspiration from Weitz's groundbreaking work \cite{Weitz} on the hard-core model---a discrete repulsive spin system---the work of the second author and Perkins~\cite{michelen2020analyticity} proved that there is no phase transition for $\lambda < e/C_\phi$.  This was extended further in \cite{mp-CC} which showed that $C_\phi$ may be replaced with the potential-weighted connective constant $\Delta_\phi$.  
	
	In practice, various Markov chain Monte Carlo algorithms are used to sample from Gibbs point processes, both approximately and exactly.  There is a large literature dedicated to this study, see for example~\cite{metropolis1953equation,alder1957phase,Krauth,propp1996exact,haggstrom1999characterization,garcia2000perfect,moller2001review,huber2016perfect,preston1975spatial,moller1989rate,friedrich2021spectral,michelen2022strong}.
	We refer to \cite{michelen2022strong} and the references therein for more context and background on these results; we also note that \cite{michelen2022strong} demonstrates a fast mixing approximate sampling algorithm and polynomial-time \emph{random} approximation algorithm for a repulsive Gibbs point process in the regime $\lambda < e/\Delta_\phi$, the same regime in which the results of this work hold.

	When it comes to deterministic algorithms, there are several results for discrete systems. In this setting, three approaches for obtaining approximation algorithms have emerged in recent years. The first is due to Weitz~\cite{Weitz} who pioneered an approach based on a notion of \emph{correlation decay} (strong spatial mixing) related to the absence of phase transitions. The second is the \emph{interpolation method} introduced by Barvinok~\cite{barvinok2016combinatorics}--a framework for showing that under a zero-freeness assumption, one can approximate the logarithm of a polynomial using a small number of Taylor coefficients (see also \cite{patel2017deterministic} for an important extension of this method). The third is a method based on the \emph{cluster expansion} in statistical physics pioneered by Helmuth, Perkins and Regts~\cite{helmuth2020algorithmic} that is closely related to Barvinok's method.

	For Gibbs point processes, to our knowledge the only rigorous result giving a deterministic algorithm is that of Friedrich, G\"obel, Katzmann, Krejca and Pappik \cite{friedrich2022algorithms}. They show that for the hard-spheres model, one may approximately compute the partition function in quasipolynomial time for $\lambda < e/C_\phi$.  The approach of \cite{friedrich2022algorithms} works by showing that one may approximate the partition function of the hard-sphere model with the partition function for the hard-core model on a graph given by discretizing Euclidean space with a small mesh.  After this approximation is in place, an application of Weitz's method \cite{Weitz} provides an algorithm.  We note also that \cite{friedrich2022algorithms} applies not only to the hard-sphere model, but also to a class of multi-type Gibbs point processes with hard constraints, an example of this more general class being the Widom-Rowlinson model.  Additionally, while both the algorithm of \cite{friedrich2022algorithms} and Theorem \ref{th:main} have quasipolynomial runtime, the algorithm of \cite{friedrich2022algorithms} runs in time $\exp(O( \log^2 (|\Lambda_n|/\eps))$ rather than our runtime of $\exp(O( \log^3 (|\Lambda_n|/\eps)).$ 
	
	{To our knowledge, Theorem \ref{th:main} marks the first approximation algorithm of any kind for the partition function for stable, non-repulsive potentials.}

	\subsection{Our approach}
	Our approach is via Barvinok's interpolation method combined with use of the cluster expansion.  The cluster expansion---also called the Meyer series---is a combinatorial description of the Taylor coefficients of $\log Z_S(\lambda)$.  In particular for bounded, measurable $S\subset \R^d$ and {$|\lambda| < (e^{1+B} \hat C_\phi)^{-1}$} we have 
	\begin{align}
		\log Z_S(\lambda) &= \sum_{k \geq 1} \frac{\lambda^k}{k!} C_k(S) \\
		C_k(S) &= \sum_{G \in \mathcal{G}_k} \int_{S^k} \prod_{\{i,j\} \in E(G)} (e^{-\phi(x_i-x_j)} - 1) \, d \bx \label{eq:C_k-def}
	\end{align}
	where $\mathcal{G}_k$ is the set of connected labeled graphs with $k$ vertices and $E(G)$ is the set of edges in a graph $G$ (see, e.g.,  \cite{procacci2017convergence} for this and similar expansions). 
	
	The algorithmic significance of the cluster expansion stems from two observations. The first is that in order to compute $\frac{1}{|S|}\log Z_S(\lambda)$ up to an additive error of $\eps$, it suffices to compute the first $O(\log (|S|/\eps))$ coefficients $C_k(S) $. This simple but crucial observation is at the heart of Barvinok's interpolation method and algorithmic applications of the cluster expansion. The second is that the description of the Taylor coefficients as a sum over \emph{connected} graphs allows us to approximate them efficiently. 	 
	
	At first sight, it seems that the use of the cluster expansion restricts the range of $\lambda$ for which we can obtain algorithms to the interval {$(0, ( e^{1+B} \hat C_\phi)^{-1})$} where the cluster expansion is known to converge. Moreover, the work of Groeneveld~\cite{groeneveld1962two} and Penrose \cite{penrose1963convergence} implies that for a repulsive potential the radius of convergence of the cluster expansion is at most $1/C_\phi$ (see Remark 3.7 in~\cite{nguyen2020convergence}).  In order to use the cluster expansion to produce an algorithm throughout the zero-free region (which {for repulsive potentials is known to include} values of $\lambda$ \emph{outside} of the radius of convergence), we use an idea of Barvinok and apply a well-chosen polynomial to map between the zero-free region and the disk. This is handled in Section \ref{sec:conformal}.  
	
	The main technical contribution of this work is to approximate the coefficients of the cluster expansion~\eqref{eq:C_k-def} in quasipolynomial time.  We note that the problem of approximating coefficients of the cluster expansion was pointed out in \cite{friedrich2022algorithms} as a central obstacle to obtaining an efficient deterministic algorithm that works entirely in the continuum.  To approximate the coefficients $C_k(S)$, the main challenge is that for each graph $G \in \mathcal{G}_k$, the integrand in \eqref{eq:C_k-def} need not be well-behaved; in particular, it need not be Lipschitz.  To handle this, we break each such term into many subsequent terms, each of which will give us an integral of a function that is Lipschitz over its support.  We then approximate each of these integrals by taking a sufficiently fine mesh and comparing the integral to a weighted sum over this mesh.
	 {Our assumption that the potential $\phi$ has finite range allows us to restrict our attention to approximating integrals over regions whose volume is independent of $n$. It would be interesting to extend the results of this paper to include stable potentials of infinite range.}

	\section{Approximating coefficients in the cluster expansion}
	Throughout this section, we fix a {stable} potential $\phi$ with {with stability constant $B$} that satisfies Assumption~\ref{assumption:potential}. Recall that $\Lambda_n:= [-n,n]^d \subset \R^d$. 
Our goal in this section is to provide an approximation algorithm for $C_k(\Lambda_n)$ (as defined at~\eqref{eq:C_k-def}).  In the next section we show how to use these approximate coefficients to arrive at an approximation of the partition function $Z_{\Lambda_n}(\lambda)$. 
	\begin{prop}\label{pr:compute-cluster}
		Let $\phi$ be a {stable} potential satisfying Assumption \ref{assumption:potential}.  There are constants $C, c > 0$ depending only on $\phi$ and the dimension $d$ so that the following holds.   For each $\eps \in (0,1)$ and $k,n \in \N$ we may approximate $C_k(\Lambda_n)/|\Lambda_n|$ up to an additive error of $\eps$ in time at most $C |\Lambda_n| \eps^{-dk} e^{ck^3}$.
	\end{prop}

	Note that Proposition \ref{pr:compute-cluster} makes no use of a zero-freeness or correlation decay type assumption.  
	
	In order to prove Proposition \ref{pr:compute-cluster}, we will approximate each summand in \eqref{eq:C_k-def}.   Further, we will break up the term for each graph into the terms that are Lipschitz on their support.  The main work of this section is Lemma \ref{lem:discretize} below, which provides an approximation for each of these Lipschitz terms.  Before stating this precisely, we make a few definitions. 
	
	\begin{defn}
		{Let $\Delta_1,\ldots, \Delta_\ell$ be defined as in Assumption~\ref{assumption:potential}}.
		\begin{enumerate}
			\item
			Given a connected graph $G$ on $k$ vertices, an \emph{edge-labelling} is a function $\sigma: E(G) \to {= \{1, \ldots, l\}}$.  
			\item
			For each labelling $\sigma$ and edge $\{i,j\} \in E(G)$, define the function $f_{i,j}^\sigma : (\R^d)^k \to {\mathbb R}$ via 
			\[f_{i,j}^{\sigma}(\bx) := (e^{-\phi(x_i - x_j)} - 1) \one\{ x_i - x_j \in \Delta_{\sigma(\{i,j\})}  \}\,.\]
			.
			\item
			Define $f^{\sigma}$ via 
			\[ f^{\sigma}(\bx) := \prod_{\{i,j\} \in E(G)} f_{i,j}^{\sigma}(\bx)\, .\]
			\item Define $A_\sigma \subset (\R^d)^k$ via 
			\[
			A_\sigma := \{\bx \in (\R^d)^k : x_i - x_j \in \Delta_{\sigma(\{i,j\}) }\text{ for all }\{i,j\} \in E(G)  \}\, ,
			\]
			and note that $\supp(f^\sigma) \subseteq A_\sigma$.
		\end{enumerate}
	\end{defn}

Recall that $\mathcal{G}_k$ is the set of connected labeled graphs with $k$ vertices. Given a graph $G \in \mathcal{G}_k$ we have
	\begin{align}\label{eq:fsigint}
		\int_{\Lambda_n^k} \prod_{\{i,j\} \in E(G)} (e^{-\phi(x_i - x_j)} - 1) \, d \bx = \sum_{\sigma} \int_{\Lambda_n^k}f^{\sigma}(\bx) d \bx\,,
	\end{align} 
	where the sum ranges over all edge-labellings $\sigma$ of $G$. We note that there are $\ell^{|E(G)|} \leq \ell^{k^2 /2}$ such labellings, and there are at most $2^{k^2/2}$ connected graphs on $k$ vertices. 
	
	With this setup in mind, in order to approximate $C_k(\Lambda_n)$, it is sufficient to approximate the integral of $f^\sigma$ for each labelling $\sigma$ of each graph $G \in \mathcal{G}_k$.  In the next subsection, we will demonstrate such an approximation:

	\begin{lemma}\label{lem:discretize}
		Let $n,k\in \N$. For any $G \in \mathcal{G}_k$, labelling $\sigma$, and $\delta < c_{d,R}$, there is a set of points $\mathcal{S}_{\sigma,\delta} \subset \Lambda_n^k$ with $|\mathcal{S}_{\sigma,\delta}| = O(|\Lambda_n|(2R)^{d(k-1)}\delta^{-dk} )$ so that  we have  $$\left|\int_{\Lambda_n^k} f^\sigma(\bx) \,d\bx -  \sum_{\bx \in \mathcal{S}_{\sigma,\delta}} \delta^{dk} f^\sigma(\bx) \right| =  O(k^2{e^{Bk^2}} \delta  |\Lambda_n|  (2R)^{d(k-1)} )\,. $$
		Further, the set $\cS_{\sigma,\delta}$ may be computed in time $O(k^2|\Lambda_n|(2R)^{d(k-1)}\delta^{-dk} )$.
	\end{lemma}

Proposition \ref{pr:compute-cluster} follows quickly from here.

	\begin{proof}[Proof of Proposition \ref{pr:compute-cluster}]
		Note that by~\eqref{eq:fsigint}
		\[
		C_k(\Lambda_n)=  \sum_{G \in \mathcal{G}_k}  \sum_{\sigma} \int_{\Lambda_n^k}f^{\sigma}(\bx) d \bx\,.
		\]
		Taking $\delta = \eps \exp(-Ck^2)$ in Lemma~\ref{lem:discretize} with $C$ sufficiently large as a function of $R,d,\ell$ {and $B$} we may approximate $C_k(\Lambda_n)/|\Lambda_n|$ up to additive error $\eps$ in time $O(k^2 {e^{Bk^2}}|\Lambda_n| (2R)^{d(k-1)}\delta^{-dk} )= O(|\Lambda_n| \eps^{-dk} e^{ck^3})$ where $c$ is a constant depending only on $R,d$ and $\ell$ {and $B$}.
	\end{proof}

	\subsection{Approximating the integral over a labelling}
	
	In this section, we prove Lemma~\ref{lem:discretize}.  Fix $G\in\mathcal{G}_k$ and set $S = \Lambda_n$.

	Recall from Assumption~\ref{assumption:potential} that the sets $\Delta_j$ are defined via $\Delta_j = K_j \setminus K_{j - 1}$.  For a given $\gamma \geq 0$, define 
	\[
	\Delta_j^{(\gamma)} := (1 - \gamma)K_j \setminus (1 + \gamma)K_{j - 1}\, .
	\]  
	
	Define the set $U_\gamma$ via $$U_\gamma := \{\bx  \in S^k : x_i - x_j \in \Delta_{\sigma(\{i,j\})}^{(\gamma)} , \text{ for all } \{i,j\} \in E(G)  \} \cap \supp(f^\sigma)\, .$$

Before continuing, let us motivate the definition of $U_\gamma$. We will show that if $\bx\in U_\gamma$, then $f^\sigma$ is Lipschitz on a small neighbourhood of $\bx$ (see Lemmas~\ref{lem:deltint} and~\ref{lem:sigmalip} below). This will allow us to approximate the integral $\int_{U_\gamma} f^\sigma(\bx)\,d\bx$ by a sum over a mesh. Indeed, we define 

\[\mathcal{S}_{\sigma,\delta} := U_\gamma \cap ((\delta\Z)^d)^k\,, \]
and prove the following lemma. 

\begin{lemma}\label{lem:lipschitz-approx}
		$$\left|\int_{U_\gamma} f^\sigma(\bx)\,d\bx - \sum_{\bx \in \cS_{\sigma,\delta}} \delta^{dk} f^\sigma(\bx)\right| = O(k^2 {e^{Bk^2}}\delta |S| (2R)^{d(k-1)} )\,. $$
	\end{lemma}
To complete the proof of Lemma~\ref{lem:discretize} we show that $\int_{U_\gamma} f^\sigma(\bx)\,d\bx$  closely approximates $\int_{\Lambda_n^k} f^\sigma(\bx) \,d\bx$ and that the sum in Lemma~\ref{lem:lipschitz-approx} can be computed efficiently. We turn our attention first to the latter task.

	\begin{lemma}
		The set $\mathcal{S}_{\sigma,\delta}$ may be enumerated in time $O(k^2|S| (2R)^{d(k-1)}) \delta^{-dk})$ and satisfies $|\mathcal{S}_{\sigma,\delta}| = O(|S| (2R)^{d(k-1)}) \delta^{-dk})$.
	\end{lemma} 
	\begin{proof}
		Fix a spanning tree $T$ of $G$ and let
		\[
		\cT = \{\bx \in S^k \cap ((\delta \Z)^d)^k : \|x_i - x_j\|_\infty \leq 2R \text{ for all }\{i,j\} \in E(T) \}\, .
		\]
		
		First observe that $\cS_{\sigma,\delta}\subseteq \cT$, since if $\bx\in U_\gamma$ then $x_i-x_j\in \Delta^{(\gamma)}_{\sigma(\{i,j\})}\subseteq [-R,R]^d$ for all $ \{i,j\} \in E(G)$. In particular,  $\|x_i - x_j\|_\infty \leq 2R$ for all $ \{i,j\} \in E(T)$.

		Note that $|\mathcal{T}| = O(|S|(2R)^{d(k-1)}\delta^{-dk} )$ and that $\mathcal{T}$ may be enumerated in time  $O(|S|(2R)^{d(k-1)}\delta^{-dk} )$ as well.  Further, for each point in $\cT$, we may check membership in $\cS_{\sigma,\delta}$ in time $O(k^2)$, completing the lemma.
	\end{proof}

To show that $\int_{U_\gamma} f^\sigma(\bx)\,d\bx$  closely approximates $\int_{S^k} f^\sigma(\bx) \,d\bx$ it will be enough to show that the measure of 

\[W_\gamma = S^k \cap \supp(f^\sigma) \setminus U_\gamma\,,\] 
is small.

\begin{lemma}\label{lem:Wbd}
		For $\gamma \leq 1/d$ we have $|W_\gamma| \leq 4|S| |E(G)| (2R)^{d(k-1)} d \gamma $. 
	\end{lemma}
	\begin{proof}
		For each edge $\{a,b\} \in E(G)$ define $$W_{\gamma,a,b} = \{\bx \in S^k:  x_i - x_j \in \Delta_{\sigma(\{i,j\})} \text{ for all } \{i,j\} \in E(G), \text{ and }x_a - x_b \notin \Delta_{\sigma(\{a,b\})}^{(\gamma)}  \}\,. $$
		Note then that $W_\gamma \subseteq \cup_{\{a,b\}\in E(G) } W_{\gamma,a,b}$.  Thus, it is sufficient to prove \begin{equation} \label{eq:W-gab-need}
			|W_{\gamma,a,b}| \leq 4 |S| (2R)^{d(k-1)} d\gamma \,.
		\end{equation}
		
		To see \eqref{eq:W-gab-need}, fix a spanning tree $T$ of $G$ so that $\{a,b\} \in E(T)$.  We then have that 
		\[
		W_{\gamma,a,b} \subseteq \{\bx \in (\R^d)^k : x_1 \in S, x_i - x_j \in [-R,R]^d \text{ for all }\{i,j\} \in E(T), x_a - x_b \in \Delta_{\sigma(\{a,b\})} \setminus \Delta_{\sigma(\{a,b\})}^{(\gamma)}   \}\,.\]
		This provides a bound of $$|W_{\gamma,a,b}| \leq |S|  (2R)^{d(k-2)} | \Delta_{\sigma(\{a,b\})} \setminus \Delta_{\sigma(\{a,b\})}^{(\gamma)}|\,. $$
		
		Write $j = \sigma(\{a,b\})$ for notational simplicity. By convexity of the sets $K_j$ and $K_{j-1}$ we have $(1-\gamma)K_j\subset K_j$ and $K_{j-1}\subset (1+\gamma)K_{j-1}$. We may therefore bound
		 \begin{align}\label{eqDeltaShellVol}
			| \Delta_{j} \setminus \Delta_{j}^{(\gamma)}| &\leq |K_j| - |(1-\gamma)K_j| + |(1 + \gamma)K_{j-1}| - |K_{j-1}| \nonumber\\
			&= |K_j| \left(1 - (1 - \gamma)^d \right) + |K_{j-1}|\left((1 + \gamma)^d - 1\right) \nonumber \\
			&\leq (2R)^d \cdot 2d\gamma +  (2R)^d \cdot 2d\gamma \,. 
		\end{align}
		Combining the two previous displayed equations shows \eqref{eq:W-gab-need} and completes the lemma.
	\end{proof}
Assuming Lemma~\ref{lem:lipschitz-approx}, the proof of Lemma~\ref{lem:discretize} now follows quickly.
	\begin{proof}[Proof of Lemma \ref{lem:discretize}]
		Bound \begin{align}\label{eqFormerLatter}
			\left|\int_{S^k} f^\sigma(\bx) \,d\bx -  \sum_{\bx \in \mathcal{S}_{\sigma,\delta}} \delta^{dk} f^\sigma(\bx) \right| &\leq \left|\int_{S^k \setminus U_\gamma} f^\sigma(\bx) \,d\bx\right| +  \left|\int_{U_\gamma} f^\sigma(\bx)\,d\bx - \sum_{\bx \in \cS_{\sigma,\delta}} \delta^{dk} f^\sigma(\bx)\right|\,.
		\end{align}
	{Since $|e^{-\phi(x)} - 1| \leq e^{2B}$, we have that $|f^\sigma| \leq e^{Bk^2 }$.}
		{The first term on the RHS of~\eqref{eqFormerLatter}}
	may {thus} be bounded by ${e^{Bk^2 }}|W_\gamma|$ which we bound using Lemma~\ref{lem:Wbd}, and the latter {term} may be bounded by Lemma \ref{lem:lipschitz-approx}.
	\end{proof}
	
	It remains to prove Lemma~\ref{lem:lipschitz-approx}. As discussed, a key step will be to show that if $\bx\in U_\gamma$, then $f^\sigma$ is Lipschitz on a small neighbourhood of $\bx$. This is carried out in the following two lemmas. 
	
		\begin{lemma}\label{lem:sigmalip}
		The function $f^{\sigma}$ is $(2{e^{Bk^2}} L|E(G)|)$-Lipschitz on $A_\sigma$.
	\end{lemma}
	\begin{proof}
		Let $\bx, \by \in A_\sigma$. Then by definition
		\begin{align}\label{eq:lipdef}
			|f^{\sigma}(\bx)-f^{\sigma}(\by)|= \left|\prod_{\{i,j\} \in E(G)} (e^{-\phi(x_i - x_j)} - 1) - \prod_{\{i,j\} \in E(G)} (e^{-\phi(y_i - y_j)} - 1)  \right|\, .
		\end{align}
		To bound~\eqref{eq:lipdef} we use the inequality
		\[
		\left |\prod_{i=1}^mz_i -  \prod_{i=1}^mw_i\right|\leq {T^{m-1}}\sum_{i=1}^m|z_i-w_i|\, ,
		\]
		which holds whenever $|z_i|\leq {T}, |w_i|\leq {T}$ for all $i$. 
		This yields
		\begin{align*}
			|f^{\sigma}(\bx)-f^{\sigma}(\by)| &\leq {e^{B k^2}} \sum_{\{i,j\} \in E(G)} \left| e^{-\phi(x_i - x_j)} - e^{-\phi(y_i - y_j)} \right| \\
			& \leq  {e^{B k^2}} |E(G)| L \max_{\{i,j\} \in E(G)} \|(x_i-y_i) -(x_j-y_j) \|_2 \\
			&\leq 2 {e^{Bk^2}} |E(G)|L  \|\bx-\by\|_2\, ,
		\end{align*}
		where we used that $x\mapsto \exp(-\phi(x))$ is $L$-Lipschitz on each $\Delta_j$ by Assumption~\ref{assumption:potential}.
	\end{proof}

	Set \begin{equation}\label{eq:gamma-choice}
		{\gamma := 2R  \delta\,.}
	\end{equation} 
	
	Since we are working in the context of Lemma \ref{lem:discretize}, we may also assume that $\gamma \leq 1/d$ by taking $\delta$ small enough as a function of $d$ and $R$.

	\begin{lemma}\label{lem:deltint}
		Let $\gamma$ be as in \eqref{eq:gamma-choice}. If $\bx\in U_\gamma$, then 
		\[
		B_\infty(\bx; \delta)\subseteq A_\sigma\, ,
		\]
		where $B_\infty(\bx; \delta)$ denotes the open $\ell_\infty$ ball of radius $\delta$ centred at $\bx$.
	\end{lemma}
	\begin{proof}
		Let $\bx\in U_\gamma\subseteq \supp(f^{\sigma})$.
	Suppose that $\bv\in(\R^d)^k$ is such that $\|\bv\|_\infty<\delta$. It suffices to show that $\bx+\bv\in A_{\sigma}$. Let $\{i,j\}\in E(G)$ and let $t=\sigma(\{i,j\})$. Our task is to verify that
		\begin{align}\label{eq:xvDelta}
			x_i - x_j + v_i - v_j \in  \Delta_{t}\, .
		\end{align}

		Since $\bx\in U_\gamma$ we have, by definition,
		\[
		x_i - x_j\in \Delta^{(\gamma)}_{t}\, .
		\]
		Moreover $\|v_i - v_j\|_\infty\leq 2\delta$. 
		The statement~\eqref{eq:xvDelta} then follows from the following claim. 
		\begin{claim}\label{claim:scaledDelta}
			If $\by\in \Delta^{(\gamma)}_{t}$, then $B_\infty(\by; 2\delta)\subseteq  \Delta_{t}.$
		\end{claim}

{		\begin{proof}
Recall that 
			\[
			\Delta_t= K_t \setminus K_{t-1} \quad\text{ and }\quad \Delta_t^{(\gamma)} = (1 - \gamma)K_t \setminus (1 + \gamma)K_{t - 1}\, .
			\]
			We will show that $B_\infty(\by; 2\delta)\subseteq K_t$. The argument to show that $B_\infty(\by; 2\delta)\subseteq K_{t-1}^c$ is analogous. 
			We let $\|\cdot\|$ denote the norm associated to the centrally symmetric convex set $K_t$, that is for $\mathbf u \in \R^{d}$, 
			\[
			\| \mathbf u \| := \inf \{\theta> 0 : \mathbf u \in \theta K_t \}\, .
			\]
			We note that since $[-1/R,1/R]^d\subseteq K_t \subseteq [R,R]^d$ by assumption, we have
			\begin{align}\label{eqEquivNorms}
			\frac{1}{R} \|\mathbf u\|_\infty \leq  \|\mathbf u\| \leq R\|\mathbf u\|_\infty\, .
			\end{align}
		        With these observations in hand, we note that since $\by \in \Delta^{(\gamma)}_{t}\subseteq (1-\gamma)K_t$, we have $\|\by\|\leq (1-\gamma)$.
	Suppose now that $\|\bz\|_\infty< 2\delta$, then by the triangle inequality,~\eqref{eq:gamma-choice} and~\eqref{eqEquivNorms}
	\[
	\|\by +\bz\|\leq \|\by\| + \|\bz\| < 1-\gamma + 2R\delta=1\, .
	\]
	In other words $\by +\bz\in K_t$ as desired. 
		\end{proof}
}
	
		\noindent Applying Claim \ref{claim:scaledDelta} verifies \eqref{eq:xvDelta}.
	\end{proof}

	With the previous two lemmas in hand, we are now in a position to prove Lemma~\ref{lem:lipschitz-approx}.

	\begin{proof}[Proof of Lemma~\ref{lem:lipschitz-approx}]
		Given $\bx\in (\R^{d})^k$, let $r(\bx)$ denote the point in $((\delta\Z)^d)^k$ closest to $\bx$ in $\ell_\infty$ distance (breaking ties arbitrarily). 
		Note that if $\bx\in U_\gamma$ then $r(\bx)\in A_\sigma$ by Lemma~\ref{lem:deltint}.
		It follows from Lemma~\ref{lem:sigmalip} that 
		\begin{align}\label{eq:f-round}
			\left|\int_{U_\gamma} f^\sigma(\bx) -  f^\sigma(r(\bx))\,d\bx\right| \leq 2{e^{Bk^2}}LE(G)\sqrt{dk}\delta |U_\gamma|= O( {e^{Bk^2}}|E(G)|\delta  |S| (2R)^{d(k-1)} ) \, .
		\end{align}

		Define \begin{equation*}
			U_\gamma' := \bigcup_{\bx \in \cS_{\sigma,\delta}} B_\infty(\bx,\delta/2)
		\end{equation*}
		and note that \begin{equation}\label{eq:rounded-to-sum}
			\int_{U_\gamma'} f^\sigma(r(\bx))\,d\bx =  \sum_{\bx \in \cS_{\sigma,\delta}} \delta^{dk} f^\sigma(\bx)\,.
		\end{equation}

		 Using the bound $|f^\sigma| \leq {e^{Bk^2}}$ gives \begin{equation}\label{eq:Ugprime-bound}
			\left| \int_{U_\gamma} f^{\sigma}(r(\bx))\,d\bx -  \int_{U_\gamma'} f^{\sigma}(r(\bx))\,d\bx   \right| \leq {e^{Bk^2}} |U_\gamma' \triangle U_\gamma| \, .
		\end{equation}
{Combining lines \eqref{eq:f-round}, \eqref{eq:rounded-to-sum} and \eqref{eq:Ugprime-bound}, it suffices to show that $|U_\gamma' \triangle U_\gamma|=O\left(|E(G)| \delta |S| (2R)^{d(k-1)} \right)$  to complete the lemma.
To this end note that 
\[
U_\gamma' \triangle U_\gamma \subseteq X_\gamma:= \{\bx \in (\R^d)^k : d_\infty(\partial U_\gamma, \bx) \leq \delta/2  \}\, .
\] 
If $\by\in \partial U_\gamma$, then $y_a-y_b\in  \partial \Delta_{\sigma(\{a,b\})}^{(\gamma)} $ for some $\{a,b\}\in E(G)$. 
Thus, if $\bx \in X_\gamma$, then 
\begin{align*}
d_\infty( x_a-x_b,  \partial \Delta_{\sigma(\{a,b\})}^{(\gamma)} )\leq \delta
\end{align*}
 for some $\{a,b\}\in E(G)$. Arguing as in Lemma~\ref{lem:Wbd}, it suffices to show that for $t\in [\ell]$,
 \[
 |\{z: d_\infty( z,  \partial \Delta_{t}^{(\gamma)} )\leq \delta\}| = O((2R)^d\delta)\, .
 \]
 
 Suppose then that $d_\infty( z,  \partial \Delta_{t}^{(\gamma)} )\leq \delta$.
Since,
 \[
 \partial\Delta_{t}^{(\gamma)}  \subseteq \partial( (1-\gamma)K_t) \cup  \partial ((1+\gamma)K_{t-1})\, ,
 \]
 let us suppose first that 
 \begin{align}\label{eqDistToBdry}
 d_\infty(z,   \partial( (1-\gamma)K_t)  )\leq \delta\, .
 \end{align}
 As in the proof of Claim~\ref{claim:scaledDelta}, let $\|\cdot\|$ denote the norm associated to $K_t$. Then by~\eqref{eqDistToBdry}, there exists $p, q\in \R^d$ such that $z=p+q$, $\|p\|=(1-\gamma), \|q\|_\infty\leq \delta$. By the triangle inequality,~\eqref{eq:gamma-choice} and~\eqref{eqEquivNorms}
 \[
 \|z\|\leq \|p\| + \|q\|\leq 1-\gamma + R\delta< 1.
 \]
 Similarly $ \|z\|>  1-2\gamma$. In other words,
 \[
 z\in K_t \backslash (1-2\gamma)K_t\, .
 \]
 If instead  $d_\infty(z,   \partial( (1+\gamma)K_{t-1})  )\leq \delta$, then by a similar argument we have that  $z\in (1+2\gamma)K_{t-1} \backslash K_{t-1}$. By the same calculation as in~\eqref{eqDeltaShellVol}, we conclude that
 \[
  |\{z: d_\infty( z,  \partial \Delta_{t}^{(\gamma)} )\leq \delta\}| \leq | K_t \backslash (1-2\gamma)K_t| +  |(1+2\gamma)K_{t-1} \backslash K_{t-1}|= O((2R)^d\delta)\, , 
 \]
 as desired. }	
 \end{proof}

	\section{Reducing to cluster expansion coefficients} \label{sec:conformal}
	
	Here we show that one can approximate $\frac{1}{|\Lambda_n|}\log Z_{\Lambda_n}(\lambda)$ using approximations of the coefficients of the cluster expansion.  A slight nuisance is that $\lambda$ need not lie in the radius of convergence of the cluster expansion; in particular, it is known ~\cite{groeneveld1962two,penrose1963convergence} that for a repulsive potential the radius of convergence of the cluster expansion is at most $1/C_\phi$ (see remark 3.7 in~\cite{nguyen2020convergence}). 
	To get around this issue, we will use an idea of Barvinok and precompose our function $f = |\Lambda_n|^{-1}\log Z_{\Lambda_n}$ with a well-chosen polynomial map sending the unit disk into our zero-free region.  This will result in a function that is analytic in the disk, and the appearance of the polynomial will end up providing only a polynomial fuzz to the efficiency of our algorithms.  This approach is outlined in Barvinok's monograph \cite[pg.\ 22]{barvinok2016combinatorics} on partition functions; we isolate and prove an abstract version of this idea here.  
	
	Define $\cN_{\gamma} := \{z \in \C : d(z,[0,1]) < \gamma  \}$ for all $\gamma \in (0,1)$.  
	
	\begin{theorem}\label{th:reduction-cluster} Let  $\gamma \in (0,1)$.
		Suppose $f$ is analytic in $\cN_\gamma$ and $|f(z)| \leq 1$ for all $z \in \cN_{\gamma}$.  Then there is a constant $C = C_\gamma > 0$ depending only on $\gamma$ so that the following holds for all $\eps \in (0,1/2)$.  If one can approximate each of $f^{(j)}(0)/j!$ up to an additive error of $\eps^C$ for all $j = 0,1,\ldots,C \log (1/\eps)$ in time at most $T$, then one can approximate $f(1)$ up to an additive error of $\eps$ in time $ C T \log(1/\eps) + \log^C(1/\eps)\,.$ 
	\end{theorem}

	We let $D=\{z\in \C: |z|<1\}$ denote the open unit disk in $\C$. Our first step is finding a polynomial to map $\overline D$ into the region $\cN_\gamma$ so that some point in $D$ is mapped to $1$.  While there are many ways to find such a polynomial (for example, \cite[Lemma 2.2.3]{barvinok2016combinatorics}  gives an explicit construction) the form of the polynomial is not important for us here, and so we state the relevant properties:
	
	\begin{lemma}\label{lem:get-polynomial}
		For each $\gamma > 0$ we may find a polynomial $\Phi = \Phi_\gamma$ so that $\Phi(\overline{D}) \subset \cN_\gamma$, $\Phi(0) = 0$ and there is a point $z_1\in D$ so that $\Phi(z_1) = 1$.
	\end{lemma}
	
	The point now will be to work with the function $g: \overline{D} \to \C$ defined by $g = f \circ \Phi$.  Under the assumptions of Theorem \ref{th:reduction-cluster}, such a function $g$ is analytic in $D$ and $|g(z)| \leq 1$ for all $z \in \overline{D}$.  The Cauchy integral formula will imply that we can approximate $g$ by its Taylor polynomial provided we are not near the boundary:
	\begin{fact}\label{fact:cauchy}
		Let $g$ be analytic in $D$ with $|g(z)| \leq 1$ for all $z \in \overline{D}$.  Then for all $z \in D$ and $k \in \N$ we have 
		$$\left|g(z) - \sum_{j = 0}^{k-1} \frac{g^{(j)}(0) }{j!} z^j \right| \leq \frac{|z|^{k}}{1 - |z|}\,.$$
	\end{fact}
	\begin{proof}
		Bound $$\left|g(z) - \sum_{j = 0}^{k-1} \frac{g^{(j)}(0) }{j!} z^j \right| \leq \sum_{j \geq k} \left|\frac{g^{(j)}(0) }{j!} \right|\cdot |z|^j$$
		and note that by Cauchy's integral formula we have $$\left|\frac{g^{(j)}(0) }{j!}\right|\leq 1\,.$$  Summing over $j \geq k$ completes the proof.
	\end{proof}

	Let $z_1$ be as in Lemma~\ref{lem:get-polynomial}. It follows that in order to approximate $f(1) = g(z_1)$ up to an additive error of $\eps$, it is sufficient to approximate the first $k = C_\gamma \log(1/\eps)$ terms $g^{(j)}(0)/j!$ up to an error of $\eps/(2k)$ each.  To do so, we will relate the derivatives of $g$ to those of $f$.  Iterating the chain rule gives the classical Fa\`a di Bruno formula:

	\begin{fact}[Fa\`a di Bruno's formula]\label{fact:fdB}
		Let $\Phi$ and $f$ be analytic at $0$ with $\Phi(0) = 0$ and set $g = f\circ \Phi$.  Then for each $n \in \N$ we have \begin{equation}
			g^{(n)}(0) = \sum_{k = 1}^n f^{(k)}(0) B_{n,k}(\Phi'(0),\Phi''(0),\ldots,\Phi^{(n-k+1)}(0))
		\end{equation}
		where $B_{n,k}(x_1,x_2,\ldots,x_{n-k+1})$ are the Bell polynomials given by \begin{equation}\label{eq:bell-def}
			B_{n,k}(x_1,x_2,\ldots,x_{n-k+1}) = \sum \frac{n!}{j_1! j_2! \cdots j_{n-k+1}! }\left(\frac{x_1}{1!} \right)^{j_1}\left(\frac{x_2}{2!} \right)^{j_2}\cdots \left(\frac{x_{n-k+1}}{(n-k+1)!} \right)^{j_{n-k+1}}
		\end{equation}
		where the sum is over all sequences of non-negative integers $j_1,j_2,\ldots,j_{n - k +1}$ satisfying $\sum j_i = k$ and $\sum i j_i = n$.
	\end{fact}
	
	To make use of this fact, we will need to evaluate the Bell polynomials $B_{n,k}$ at derivatives of our polynomial $\Phi$.  A simple term-by-term bound will be good enough for our purposes.
	
	\begin{lemma}\label{lem:bell-bound}
		For each polynomial $\Phi$ there is a constant $C = C_\Phi > 0$ so that the following holds.  For all $n, k \in \N$ we may compute $B_{n,k}(\Phi'(0),\Phi''(0),\ldots,\Phi^{(n-k+1)}(0))$ in time $O(k^{C})$ and we have the bound $$\left|\frac{k!}{n!}B_{n,k}(\Phi'(0),\Phi''(0),\ldots,\Phi^{(n-k+1)}(0)) \right| \leq e^{C k}\,.$$
	\end{lemma}
	\begin{proof}
		First compute all derivatives of $\Phi$ at $0$, which takes $O_\Phi(1)$ time since $\Phi$ is a polynomial.  Set $x_j = \Phi^{(j)}(0)$ and note that $x_j = 0$ for $j > d$ where $d$ is the degree of $\Phi$.  We may thus reduce the sum in \eqref{eq:bell-def} to those for which $j_i = 0$ for all $i > d$ as all other summands are $0$.  Since the values $j_1,\ldots,j_d$ are non-negative integers of size at most $k$, we have $(k+1)^d$ choices for the sequence $(j_1,\ldots,j_d)$. Moreover the constraints $\sum j_i = k$ and $\sum i j_i = n$ imply that $n\leq dk$ and so each summand in~\eqref{eq:bell-def} can be computed in time $O_d(k^{C'})$ for some $C'=C'_d>0$. We may therefore calculate the sum~\eqref{eq:bell-def} in time $O_\Phi(k^C)$ for some $C=C_d>0$. 
		
  To see the claimed bound, apply the triangle inequality to see $$\left|\frac{k!}{n!}B_{n,k}(x_1,x_2,\ldots,x_{n-k+1})\right| \leq \sum \binom{k}{j_1,\ldots,j_d} |x_1|^{j_1} |x_2|^{j_2} \cdots |x_d|^{j_d}\,.$$ 
		Relaxing the sum to consist of sequences satisfying $\sum j_i = k$, the binomial theorem gives $$\sum \binom{k}{j_1,\ldots,j_d} |x_1|^{j_1} |x_2|^{j_2} \cdots |x_d|^{j_d} = (|x_1| + |x_2| + \ldots + |x_d|)^k \leq e^{Ck}$$ 
		for $C$ large enough with respect to $\Phi$.
	\end{proof}
	
	The proof of Theorem \ref{th:reduction-cluster} now follows from some bookkeeping.
	
	\begin{proof}[Proof of Theorem \ref{th:reduction-cluster}]
		Apply Lemma \ref{lem:get-polynomial} to find a polynomial $\Phi=\Phi_\gamma$ satisfying the hypotheses of the lemma and define $g = f \circ \Phi$.  Then by Fact \ref{fact:cauchy} we may take $k := C_\gamma \log(1/\eps)$ large enough so that $$\left|g(z_1) - \sum_{j = 0}^{k-1}\frac{g^{(j)}(0)}{j!} z_1^j  \right| \leq \eps/2\,.$$  Since $g(z_1) = f(1)$, it is thus sufficient to approximate $g^{(j)}(0)/j!$ for $j = 0,1,\ldots,k-1$ up to an additive error of at most $\delta := \eps/(2k)$.  By Fact \ref{fact:fdB} and Lemma \ref{lem:bell-bound} we may expand $$\frac{g^{(j)}}{j!}(0) = \sum_{i = 1}^j \frac{f^{(i)}(0)}{i!} B_{i,j}$$ where the coefficients $B_{i,j}$ satisfy $|B_{i,j}| \leq e^{C i} \leq e^{C k} = O(\eps^{-C})$ and may be computed in time $O(k^C) = O(\log^C(1/\eps))$.  Thus, if we take $C'$ large enough (as a function of $\gamma$) and approximate $f^{(i)}(0)/i!$ up to additive error $\eps^{C'}$, we obtain an approximation to $g(z_1) = f(1)$ up to additive error $\eps$.
	\end{proof}
	
	\section{Proof of Theorem \ref{th:main}}
	
	\begin{proof}[Proof of Theorem \ref{th:main}]
		Set $f(z):= \frac{1}{C|\Lambda_n|} \log Z_{\Lambda_n}(\lambda z)$, where $C$ is as in Assumption \ref{assumption:zero-free}.  Then $f$ satisfies the hypotheses of Theorem \ref{th:reduction-cluster}, and note that by \eqref{eq:C_k-def} we have $f^{(j)}(0)= \lambda^j C_j(\Lambda_n)/(C |\Lambda_n|)$.  If we set $\delta = \eps/(C |\Lambda_n|)$, then we are seeking to approximate $f(1)$ up to an additive error of $\delta$. By Propostion~\ref{pr:compute-cluster}, for each $j = 0,1,\ldots, O(\log(1/\delta) )$, we may compute $f^{(j)}(0)$ up to additive error $\delta^{O(1)}$ in time at most $O(|\Lambda_n| \delta^{-O(j)}e ^{O(j^3)}) = \exp(O(\log^3(1/\delta)))\,.$  Applying Theorem \ref{th:reduction-cluster} completes the proof.
	\end{proof}
	
	\section*{Acknowledgments}
	MJ is supported by a UKRI Future Leaders Fellowship MR/W007320/1.  MM supported in part by NSF grant DMS-2137623. {MR gratefully acknowledges support from the Bogazici Solidarity fund and the Institute of Mathematical Sciences, Chennai, where part of this work was carried out.} We would like to thank Will Perkins for helpful discussions. {We also thank the anonymous referee for their careful reading and helpful comments.}
	
	\bibliographystyle{alpha}
	\bibliography{sphere}
	
\end{document}